\documentclass[a4paper,12pt]{article}
\usepackage[frenchb,english]{babel}
\usepackage{geometry}
\usepackage{mathenv, xcolor}
\usepackage[T1]{fontenc} 
\usepackage{enumerate}
\usepackage[normalem]{ulem}
\usepackage{amsmath}
\usepackage{stmaryrd}
\usepackage{graphicx}
\usepackage{amsfonts}
\usepackage{amsthm}
\usepackage{geometry}
\usepackage[all]{xy}
\usepackage{verbatim}
\usepackage{algorithm2e}
\newtheorem{defi}{Definition~}[section]

\newtheorem{thm}[defi]{Theorem~}
\newtheorem{cor}[defi]{Corollary~}

\newtheorem{lem}[defi]{Lemma~}

\newtheorem{quest}{Question~}

\title{Zero forcing number, constrained matchings \\and strong structural controllability}
\author{Maguy Trefois\thanks{ICTEAM Institute - Department of Applied Mathematics, Universit\' e catholique de Louvain, Avenue Georges Lema\^itre 4 B-1348 Louvain-La-Neuve, Belgium. maguy.trefois@uclouvain.be, jean-charles.delvenne@uclouvain.be.} \and Jean-Charles Delvenne$^*$ \thanks{Center for Operations Research and Econometrics (CORE), Universit\'e catholique de Louvain, Louvain-La-Neuve, Belgium.}}
\date{May 24, 2015}

\begin{document}

\maketitle

\begin{center}
\textbf{Abstract}
\end{center}

The zero forcing number is a graph invariant introduced to study the minimum rank of the graph. In 2008, Aazami proved the NP-hardness of computing the zero forcing number of a simple undirected graph. We complete this NP-hardness result by showing that the non-equivalent problem of computing the zero forcing number of a directed graph allowing loops is also NP-hard. The rest of the paper is devoted to the strong controllability of a networked system. This kind of controllability takes into account only the structure of the interconnection graph, but not the interconnection strengths along the edges. We provide a necessary and sufficient condition in terms of zero forcing sets for the strong controllability of a system whose underlying graph is a directed graph allowing loops. Moreover, we explain how our result differs from a recent related result discovered by Monshizadeh \textit{ et al}. Finally, we show how to solve the problem of finding efficiently a minimum-size input set for the strong controllability of a self-damped system with a tree-structure.

\section{Introduction}
\label{}
The minimum rank problem has been motivated by the Inverse Eigenvalue Problem of a Graph (IEPG) \cite{fallat,hogben1}, intensively studied during the last fifteen years. In the IEPG, a simple undirected graph $G$ is considered. Such a graph defines a set $\mathcal{Q}_{su}(G)$ of real symmetric matrices whose zero-nonzero pattern of the off-diagonal entries is described by the graph: the $(i,j)$-entry (for $i \neq j$) is nonzero if and only if $\{i,j\}$ is an edge of $G$. The zero-nonzero pattern of the diagonal entries is free. Given a sequence $[\mu_1, ..., \mu_n]$ of non increasing real numbers, we ask whether there exists a matrix in the set $\mathcal{Q}_{su}(G)$ whose spectrum is $[\mu_1, ..., \mu_n].$ A first step to study this challenging problem is to compute the maximum possible multiplicity of an eigenvalue $\mu$ for a matrix in $\mathcal{Q}_{su}(G)$. This maximum is obtained thanks to the minimum rank of the graph $G$. Indeed, for any real number $\mu$, the maximum possible multiplicity of $\mu$ as an eigenvalue of a matrix in $\mathcal{Q}_{su}(G)$ is: $$|G| - mr(G),$$where $|G|$ denotes the number of vertices in $G$ (the vertex set of a graph is assumed to be finite) and $mr(G)$ refers to the minimum rank of the graph, that is: $$mr(G) = \min\{\text{rank}(A) : A \in \mathcal{Q}_{su}(G)\}.$$

Since the zero-nonzero pattern of the diagonal entries for a matrix in $\mathcal{Q}_{su}(G)$ is free, this maximum multiplicity does not depend on $\mu$.

The minimum rank of a graph is also useful in other problems like the singular graphs or the biclique partition of the edges of a graph (see \cite{fallat} for a brief description and interesting references).

Originally given for simple undirected graphs, the definition of the minimum rank of a graph has then been extended to the case of simple directed graphs and loop (undirected and directed) graphs, which are graphs allowing loops.

A \textbf{simple} (undirected or directed) graph is a graph that prohibits loops on its vertices whereas a \textbf{loop} (undirected or directed) graph is a graph that allows loops on its vertices.

In order to study the minimum rank of a graph, many graph invariants have been introduced \cite{aim1,barioli,fallat,hogben1,hogben2} . One of them is the zero forcing number, presented below.

The zero forcing number of a graph was originally defined in \cite{aim1} for simple undirected graphs and extended to loop directed graphs in \cite{barioli}. The definitions of zero forcing number for simple directed graphs and loop undirected graphs were then given in \cite{hogben2}.

Any undirected graph can be seen as a directed graph: any edge $\{i,j\}$ is split into directed edges $(i,j)$ and $(j,i)$. Consequently, as stated in \cite{hogben2}, the zero forcing number of any undirected graph is equal to the zero forcing number of its associated directed graph. Therefore, we will focus w.l.o.g. on the zero forcing number of directed graphs.

The definitions of zero forcing number and zero forcing set depend on a color change rule on the graph. This rule is slightly different in a loop directed graph or in a simple directed graph.

The \textbf{color change rule} in a \textit{loop} directed graph \cite{barioli,hogben2} is the following: suppose that any vertex of $G$ is either black or white. If exactly one out-neighbor $j$ of vertex $i$ is white (possibly $j=i$), then change the color of vertex $j$ to black.

In a \textit{simple} directed graph $G$, the color change rule is defined as follows: suppose any vertex of $G$ to be black or white. If vertex $i$ is black and vertex $j$ is the only white out-neighbor of $i$, then change the color of $j$ to black.

Hence, in a simple directed graph vertex $i$ must be black to be able to change the color of one of its out-neighbors, which is not the case in a loop directed graph.

 The color change rule is repeatedly applied to each vertex of $G$ until no more color change is possible. The  \textbf{zero forcing number} $Z(G)$ of graph $G$ is defined as the minimum number of vertices which have to be initially black so that after applying repeatedly the color change rule to $G$ all the vertices of $G$ are black. A vertex subset $S$ of $G$ with the property that if only the vertices of $S$ are initially black in $G$, then the whole graph is black after applying repeatedly the color change rule is called a \textbf{zero forcing set} of $G$. A zero forcing set of size $Z(G)$ is called a \textbf{minimum zero forcing set}.

This graph invariant allows us to compute a lower bound for the minimum rank of the graph. The definition of the minimum rank of a graph relies on the matrix family defined by the graph. This family of matrices depends on the type of the graph:

\begin{enumerate}[-]
\item if $G = (V,E)$ is a simple undirected graph, $$\mathcal{Q}_{su}(G) = \{ A \in \mathbb{R}^{|G| \times |G|} : A^T = A,  \text{for any }  i \neq j, a_{ij} \neq 0 \Leftrightarrow \{i,j\} \in E\},$$ 
\item if $G = (V,E)$ is a simple directed graph, $$\mathcal{Q}_{sd}(G) = \{ A \in \mathbb{R}^{|G| \times |G|} :\text{for any }  i \neq j, a_{ij} \neq 0 \Leftrightarrow (j,i) \in E\},$$ 
\item if $G = (V,E)$ is a loop undirected graph, $$\mathcal{Q}_{lu}(G) = \{ A \in \mathbb{R}^{|G| \times |G|} : A^T = A, \text{for any }  i,j, a_{ij} \neq 0 \Leftrightarrow \{i,j\} \in E\},$$ 
\item if $G = (V,E)$ is a loop directed graph, $$\mathcal{Q}_{ld}(G) = \{ A \in \mathbb{R}^{|G| \times |G|} : \text{for any }  i,j, a_{ij} \neq 0 \Leftrightarrow (j,i) \in E\}.$$ 
\end{enumerate} 

\textbf{Warning}: following \cite{monsh}, the matrix sets $\mathcal{Q}_{sd}(G)$ and $\mathcal{Q}_{ld}(G)$ have matrices transposed from those in the minimum rank literature \cite{barioli, hogben2}.

\medskip

The minimum rank of a graph $G$ of type $i$ with $i =$ su, sd, lu or ld is then the minimum possible rank for a matrix in $\mathcal{Q}_{i}(G)$, that is: $$mr(G) = \min\{ \text{rank}(A) : A \in \mathcal{Q}_{i}(G)\}.$$
It was proved in \cite{hogben2} that for any graph $G$, $$|G|-Z(G) \leq mr(G).$$

Moreover, if $G$ is a tree of any type (simple or allowing loops, undirected or directed), then equality holds \cite{hogben2}:  
\begin{equation}
\label{eq1}
|G|-Z(G) = mr(G).
\end{equation}

Because of the different color change rules, computing the zero forcing number of a simple (un)directed graph can \textit{not} be reduced to the computation of the zero forcing number of a loop directed graph and vice-versa. As a first result, we highlight the NP-hardness of the computation of the zero forcing number of a \textit{loop directed} graph. This completes the NP-hardness result proved by Aazami \cite{aazami} stating that computing the zero forcing number of a \textit{simple undirected} graph is NP-hard.

In the rest of the paper, we study the strong controllability of networked systems, through the zero forcing sets.

A networked system is a linear system whose dynamics is driven by a matrix $A$ with a structure provided by a graph, called the interconnection graph or the underlying graph.

Classical controllability of a networked system has been first considered in \cite{tanner} when $A$ is the Laplacian matrix of the underlying graph. This work has then continued in \cite{egerstedt, rahmani, zhang} when $A$ is the Laplacian matrix and in \cite{godsil} when $A$ is the adjacency matrix.

Classical controllability of networked systems was studied from graph theoretic arguments in \cite{camlibel, godsil, martini, rahmani, tanner, yazi, zhang}.

When $A$ is the Laplacian matrix, the minimum number of vertices that must be directly controlled by the outside controller in order to control the whole system has been studied for special graphs in \cite{nabi, parlangeli,rahmani,zhang}. The vertices directly controlled by the outside controller are called the \textbf{input set}.

Recently, zero forcing sets have been used in order to study the controllability of quantum systems \cite{burgarth3, burgarth, burgarth2, burgarth4}.

However, when the weights on the edges of the graph underlying the networked system are unknown or only partially determined, the classical Kalman condition for controllability can not be used. That is why structural controllability has been introduced.

The problem of determining a control strategy for a network of interconnected systems without exact knowledge of the interaction strengths along the edges has seen a surge of activity in the last decade, in particular regarding weak and strong structural controllability introduced in the 70's \cite{lin, mayeda}. Structural controllability takes into account only the structure of the interconnection graph, but not the interaction strengths on the edges. A system with a given interconnection graph is weakly structurally controllable from an input set $S$ if we can choose interaction strengths making the system controllable from $S$. Instead, a system with a given interconnection graph is strongly structurally controllable, or strongly controllable for short, from an input set $S$ if whatever the interaction strengths, the system is controllable from $S$.

Recently, the notion of matching has proved useful in the study of weak and strong structural controllability of a networked system. In particular, the constrained matchings in a bipartite graph are of interest.

The notion of constrained matching in a bipartite graph was originally defined in the paper of Hershkowitz and Schneider \cite{hershkowitz}. A \textbf{$t$-matching} is a set of $t$ edges such that no two edges share a vertex. Given a matching, a vertex of the graph that belongs to an edge of the matching is called a \textbf{matched vertex}, otherwise it is an \textbf{unmatched vertex}. A $t$-matching is called \textbf{constrained} if there is no other $t$-matching with the same matched vertices. A \textbf{maximum} (constrained) $t$-matching is a (constrained) $t$-matching such that there is no (constrained) $s$-matching with $s > t$.

Given a loop directed graph $G$, a bipartite graph $B_G = (V,V',E)$ was defined in \cite{hershkowitz} as follows: the sets $V$ and $V'$ are two copies of the vertex set of $G$. To avoid ambiguity, the vertices in $V$ are denoted by $i_1, ..., i_n$ and the vertices in $V'$ are denoted by $j_1, ..., j_n$. Given any vertex $i_s \in V$ and any vertex $j_t \in V'$, $\{i_s,j_t\}$ is an edge in $B_G$ if and only if there is an edge from vertex $j_t$ to vertex $i_s$ in $G$. Then $B_G$ is called the \textbf{bipartite graph associated with $G$}.
This paper is based on a one-to-one correspondance between the zero forcing sets in a loop directed graph $G$ and the constrained matchings in the bipartite graph associated with $G$. 

The minimum-size input sets from which a system is weakly structurally controllable are provided by the maximum matchings in the interconnection graph \cite{liu}, whereas the minimum-size input sets from which a system is strongly controllable are provided by some maximum constrained matchings in the bipartite graph associated with the interconnection graph \cite{chapman, olesky}.

As in \cite{chapman, liu, olesky}, we suppose the graph underlying the networked system to be a loop directed graph.

In this paper, we shed a new light on the strong controllability from the notion of zero forcing set. To our best knowledge, the equivalence between the zero forcing sets in a loop directed graph \cite{barioli, fallat, hogben2} and the constrained matchings in its associated bipartite graph \cite{chapman, golumbic, hershkowitz, mishra, olesky} has gone unnoticed. In this paper, we emphasize this equivalence and apply it in the study of the strong controllability of networked systems. Firstly, we show that testing if a system is strongly controllable from an input set $S$ is equivalent to checking if $S$ is a zero forcing set in the interconnection graph. Secondly, in the case of systems that are self-damped (i.e. the state of each vertex is influenced, among others, by itself), we show that minimum-size input sets for strong controllability are provided by the minimum zero forcing sets in the \textit{simple} interconnection graph, which is the interconnection graph without its loops. In particular,  we show that one can find in polynomial time a minimum-size input set for the strong controllability of a self-damped system with a tree structure.

The outline of this paper is as follows: in Section 2 we highlight that the computation of the zero forcing number of any loop directed graph is NP-hard. Section 3 is a statement of the strong controllability problem. Section 4 presents the main results of \cite{chapman} linking strong controllability and constrained matchings in the bipartite graph associated with the interconnection graph. Section 5 contains our results about strong controllability and zero forcing sets. In Section 6 we show how our results from the previous section differ from the related results appeared in \cite{monsh}. Section 7 contains concluding remarks.

\section{NP-hardness}

In his PhD thesis \cite{aazami}, Aazami has proved the NP-hardness of the computation of the zero forcing number in a \textit{simple undirected} graph, using a reduction from the Directed Hamiltonian Cycle problem known to be NP-hard \cite{karp}.

In this section, we complete this NP-hardness result for zero forcing by showing that the calculation of the zero forcing number of any \textit{loop directed} graph is also NP-hard. We use a result of \cite{golumbic} claiming that computing the size of a maximum constrained matching in a bipartite graph is NP-hard.

Recall that because of the different color change rules, computing the zero forcing number of a simple (un)directed graph can \textit{not} be reduced to the computation of the zero forcing number of a loop directed graph and vice versa.

\begin{defi}
Two matrices $A$ and $B$ are said to be \textbf{permutation similar} if there are permutation matrices $P_1, P_2$ such that $A = P_1BP_2$.
\end{defi}

A \textbf{zero-nonzero pattern} $\textbf{A}$ (or pattern for short) of dimension $n$ is an $n \times n$ matrix with each entry being either a star $\star$ or zero. A star $\star$ refers to a nonzero entry. 

A loop directed graph $G$ with $n$ vertices defines a zero-nonzero pattern $\textbf{A}(G)$ of  dimension $n$ as follows: the entry $a_{ij}$ of $\textbf{A}(G)$ is a star $\star$ if and only if there is a directed edge from  vertex $j$ to vertex $i$ in $G$. This pattern is called the \textbf{zero-nonzero pattern associated with $G$}.

\newpage
\begin{defi}\cite{barioli}
Let $\textbf{A}$ be a zero-nonzero pattern.  
\begin{enumerate}[-]

\item A \textbf{$t$-triangle} of $\textbf{A}$ is a $t \times t$ subpattern of $\textbf{A}$ which is permutation similar to an upper triangular pattern whose all diagonal entries are nonzero. 
\item The \textbf{triangle number of $\textbf{A}$} is the maximum size of a triangle in $\textbf{A}$. 
\item The \textbf{triangle number $tri(G)$ of a loop directed graph $G$} is the triangle number of its associated zero-nonzero pattern $\textbf{A}(G)$.
\end{enumerate}
\end{defi}

\begin{thm}\cite{barioli}
\label{thm03}
For any loop directed graph $G$, $tri(G) + Z(G) = |G|.$
\end{thm}

The following theorem proved in \cite{hershkowitz} shows the link between the constrained matchings in a bipartite graph and the triangle number.

A zero-nonzero pattern $\textbf{A}$ of dimension $n$ defines a bipartite graph $B_\textbf{A}$ whose vertex sets $V = \{i_1, ..., i_n\}$ and $V' = \{j_1, ..., j_n\}$ are two copies of $\{1, ..., n\}$ and $\{i_s,j_t\}$ is an edge in $B_\textbf{A}$ if and only if $(i_s,j_t)$-entry of \textbf{A} is a star $\star$. Then $B_\textbf{A}$ is called the \textbf{bipartite graph associated with $\textbf{A}$}.

Notice that the bipartite graph associated with a loop directed graph $G$ is equal to the bipartite graph associated with the pattern $\textbf{A}(G)$.

\begin{thm}\cite{hershkowitz}
Let $\textbf{A}$ be an $n \times n$ zero-nonzero pattern and $B_\textbf{A}$ the bipartite graph associated with $\textbf{A}$. Then the following statements are equivalent.
\begin{enumerate}[-]
\item $B_\textbf{A}$ has a constrained $n$-matching
\item $\textbf{A}$ is permutation similar to a triangular pattern with nonzero diagonal elements.
\end{enumerate}
\end{thm}

From the previous theorem, we deduce that the size of a maximum constrained matching in $B_\textbf{A}$ equals the triangle number of $\textbf{A}$. However, in \cite{golumbic} it has been proved that the computation of the size of a maximum constrained matching in a bipartite graph is NP-hard. Therefore, so it is for the triangle number of a loop directed graph.
From this result and Theorem \ref{thm03}, we have highlighted the NP-hardness of the computation of the zero forcing number of any loop directed graph.

\begin{thm}
The computation of the zero forcing number of any loop directed graph is NP-hard.
\end{thm}

\section{Strong controllability: problem formulation}

In this section, we define the strong controllability of a system underlying a loop directed graph and we present the classical questions regarding strong controllability we will consider in the rest of the paper.

\medskip

A loop directed graph $G = (V,E)$ on $n$ vertices defines the matrix set $$\mathcal{Q}_{ld}(G) = \{ A \in \mathbb{R}^{n \times n} : \text{ for any } i,j, a_{ij} \neq 0 \Leftrightarrow (j,i) \in E\}.$$ Given the vertex set $V = \{1, ..., n\}$ and a vertex subset $S = \{v_1, ..., v_m\} \subseteq V$, we define the $n \times m$ pattern $\textbf{B}(S)$ as $$[\textbf{B}(S)]_{ij} = \begin{cases} \star & \text{ if } i = v_j \\ 0 & \text{ otherwise.}\end{cases}$$ A realization of pattern $\textbf{B}(S)$ is an $n \times m$ matrix $B$ such that any $(i,j)$-entry in $B$ is nonzero if and only if the corresponding entry in $\textbf{B}(S)$ is a star $\star$. We write $B \in \textbf{B}(S)$.

A networked system whose underlying graph is $G$ is a linear system of the form: $$\dot{x}(t) = Ax(t) + Bu(t)$$ where $x(t) \in \mathbb{R}^n$ is the state vector, $u(t) \in \mathbb{R}^m$ is the outside controller, $A \in \mathcal{Q}_{ld}(G)$ and $B \in \textbf{B}(S)$ is a realization of $\textbf{B}(S)$, for some input set $S = \{v_1, ..., v_m\}$. Notice that only the vertices in $S$ are directly controlled by the outside controller. Such a system is referred to as system $(A,B)$ and is said to be controllable if the controllability matrix $$C = \left[\begin{array}{ccccc} B & AB & A^2B & ... & A^{n-1}B \end{array}\right]$$ has full rank. Concretely, a system $(A,B)$ is controllable if it can be driven from any initial state $x_0 \in \mathbb{R}^n$ to any final state $x_f \in \mathbb{R}^n$ in finite time.

\medskip

Given the loop directed graph $G$, the zero-nonzero pattern $\textbf{A} = \textbf{A}(G)$ associated with $G$ is the pattern of the matrices in $\mathcal{Q}_{ld}(G)$. Therefore, any matrix $A$ of $\mathcal{Q}_{ld}(G)$ is a realization of $\textbf{A}$. We write $A \in \textbf{A}$.

\medskip

Given an input set $S$, any system $(A,B)$ whose underlying graph is $G$ is a realization of the pair  $(\textbf{A}, \textbf{B}(S))$, meaning that $A \in \textbf{A}$ and $B \in \textbf{B}(S)$. We write $(A,B) \in (\textbf{A}, \textbf{B}(S))$. This class of systems is referred to as system $(\textbf{A}, \textbf{B}(S))$.

\medskip

Given the system $(\textbf{A}, \textbf{B}(S))$, the strong controllability seeks to know if any system $(A,B) \in (\textbf{A}, \textbf{B}(S))$ is controllable.

\begin{defi}
The system $(\textbf{A}, \textbf{B}(S))$ is strongly $S$-controllable if all systems $(A,B) \in (\textbf{A}, \textbf{B}(S))$ are controllable.
\end{defi}

In the rest of the paper, we tackle the following questions:
\begin{quest}
Given an input set $S$, is the system $(\textbf{A}, \textbf{B}(S))$ strongly $S$-controllable ?
\end{quest}
\begin{quest}
What is the minimum size of an input set $S$ making the system $(\textbf{A}, \textbf{B}(S))$ strongly $S$-controllable ? Can we find efficiently such a minimum-size input set ?
\end{quest}

Strong controllability of a networked system $(\textbf{A}, \textbf{B}(S))$ has been characterized in \cite{chapman} from the constrained matchings in the bipartite graph $B_G$ associated with the loop directed graph $G$. In the next section, we present the main results of \cite{chapman}.

\section{Strong controllability and constrained matchings}

In this section, we present the results of \cite{chapman} connecting the strong controllability of a networked system with the constrained matchings in the bipartite graph associated with the underlying loop directed graph.

\medskip

In Section 2 a bipartite graph associated with a pattern of dimension $n$ was defined. This can be extended to rectangular patterns. Let $\textbf{A}$ be an $n \times m$ zero-nonzero pattern. The bipartite graph $B_\textbf{A}$ associated with $\textbf{A}$ has vertex sets $V = \{1, ..., n\}$ and $V' = \{1, ..., m\}$. To avoid ambiguity, the elements of $V$ are denoted by $\{i_1, ..., i_n\}$ and the ones of $V'$ are denoted by $\{j_1, ..., j_m\}$. Besides, $\{i_s,j_t\}$ is an edge in $B_\textbf{A}$ if and only if $(i_s,j_t)$-entry of $\textbf{A}$ is a star $\star$.

\medskip

If the bipartite graph associated with a pattern $\textbf{A}$ has a constrained $t$-matching, then we say by abuse of language that $\textbf{A}$ has a constrained $t$-matching.

\medskip

Let $S \subseteq V$ be a vertex subset in a loop directed graph $G$. A constrained \textit{$S$-less} matching in the bipartite graph $B_G$ associated with $G$ is a constrained matching with no edges of the form $\{i,i\}$ with $i \in S$. In particular, if $S= V$, a constrained $S$-less matching in $B_G$ is called a constrained \textit{self-less} matching.

\medskip

Let $\textbf{A}$ be an $n \times m$ pattern and $S \subseteq \{1, ..., n\}$. Then $\textbf{A}(S|.)$ denotes the pattern obtained from $\textbf{A}$ by deleting the rows indexed by $S$.

\medskip

Let $\textbf{A}$ be a pattern of dimension $n$. Then $\textbf{A}_\times$ is the pattern obtained from $\textbf{A}$ by putting stars $\star$ along the diagonal. Similarly, $G_\times$ denotes the graph obtained from graph $G$ by putting a loop on each vertex of $G$.

\medskip

Throughout, $V_{loop}$ is the set of vertices with a loop in the original loop directed graph $G$ underlying the system.

\medskip

Here is the theorem from \cite{chapman} characterizing the strong controllability of a system from the constrained matchings in the bipartite graph associated with the underlying loop directed graph.

\medskip

\begin{thm}\cite{chapman}
\label{thm42}
Let $G$ be a loop directed graph on $n$ vertices with pattern $\textbf{A}$ and $S$ be an input set with cardinality $m \leq n$. System $(\textbf{A}, \textbf{B}(S))$
 is strongly $S$-controllable if and only if $\textbf{A}(S|.)$ has a constrained $(n-m)$-matching and $\textbf{A}_\times(S|.)$ has a constrained $V_{loop}$-less $(n-m)$-matching. 
 \end{thm}
 
Given an input set $S$, in order to check whether or not a system is strongly $S$-controllable, a $\mathcal{O}(n^2)$ algorithm was presented in \cite{chapman}. If the system is not strongly $S$-controllable, the algorithm computes a vertex set $\tilde{S}$  containing $S$ such that the system is strongly $\tilde{S}$-controllable. However, computing a minimum-size input set for strong controllability is a challenging problem. In that scope, the sufficient condition of the following theorem has been proved in \cite{chapman}. We deduce the necessary condition from Theorem \ref{thm42}.

\medskip

A \textit{self-damped} system is a system whose underlying graph has a loop on each vertex (the state of each vertex influences itself).

\medskip

A maximum constrained self-less $t$-matching is a constrained self-less $t$-matching such that there is no constrained self-less $s$-matching with $s > t$.

\begin{thm}\cite{chapman}
\label{thm43}
Consider a loop directed graph $G$ on $n$ vertices with pattern $\textbf{A}$ underlying a self-damped\footnote{In \cite{chapman}, the sufficient condition of this result was also stated for undamped systems, which are systems where the state of each vertex is influenced by the state of some other vertices but not itself. However, in Appendix A, we show that in an undamped system, the input set $S$ provided by the maximum constrained self-less matching may not be of minimum size, in contradiction with Corollary 11 of \cite{chapman}.} system. A vertex subset $S \subseteq V$ is a (minimum-size) input set for strong controllability of the system if and only if there is a (maximum) constrained self-less matching in $\textbf{A}_\times$ such that in the bipartite graph $(V,V',E)$ associated with $\textbf{A}_\times$, the unmatched vertex set of $V$ is $S$. 
\end{thm}
\begin{proof}
Since patterns $\textbf{A}$ and $\textbf{A}_\times$ are equal, the result is deduced from Theorem \ref{thm42}.
\end{proof}

This theorem provides a way to obtain a minimum-size input set for strong controllability in the case of a self-damped system. However, computing a maximum constrained self-less matching in a bipartite graph is a challenging problem.

\medskip

In the next section, we re-state these results in terms of zero forcing sets. On the one hand, these new statements show that testing whether or not a system is strongly $S$-controllable is equivalent to testing if $S$ is a zero forcing set in a loop directed graph. On the other hand, they show that the zero forcing sets in a simple graph provide input sets for strong controllability of a self-damped system. In particular, this result together with the existing algorithms on zero forcing provide a way to select in polynomial time a minimum-size input set for a self-damped system with a tree-structure.

\section{Strong controllability and zero forcing sets}

While Theorems \ref{thm42} and \ref{thm43} provide criteria for strong controllability in terms of constrained matchings, we provide in this section equivalent criteria in terms of zero forcing sets. On the one hand, these new statements show that given an input set $S$, testing whether or not the system is strongly $S$-controllable is equivalent to checking if $S$ is a zero forcing set in a loop directed graph. On the other hand, they show that in the case of a self-damped system, there is a one-to-one correspondance between the input sets providing strong controllability and the zero forcing sets in a simple directed graph. This result together with the existing algorithms on the zero forcing sets solves the problem of finding efficiently a minimum-size input set for strong controllability of a self-damped system with a tree-structure.

\medskip

Statements of Theorems \ref{thm42} and \ref{thm43} in terms of zero forcing sets are based on a one-to-one correspondance between the zero forcing sets in a loop directed graph $G$ and the constrained matchings in the bipartite graph associated with $G$. 

\medskip

\begin{defi}\cite{hogben2}
Let $G$ be a directed graph (simple or allowing loops).
\begin{itemize}
\item Suppose that any vertex of $G$ is either black or white. When the color change rule (cf. Section 1) is applied to vertex $i$ to change the color of  vertex $j$, we say that \textbf{$i$ forces $j$} and write $i \rightarrow j$.
\item  Given a zero forcing set of $G$, we can list the forces in order in which they were performed to color the vertices of $G$ in black. This list is called a \textbf{chronological list of forces}.
\end{itemize}
\end{defi}

Notice that given a zero forcing set, a chronological list of forces is not necessarily unique. However, uniqueness is not required here.

\medskip

Theorem \ref{thm34} below is a consequence of the following theorem proved in \cite{golumbic}.

\begin{thm}\cite{golumbic}
\label{thm2}
Let $B = (V,V',E)$ be a bipartite graph and $\mathcal{M}$ be a matching in $B$. The following assertions are equivalent:
\begin{itemize}
\item $\mathcal{M}$ is a constrained matching
\item We can order the vertices of $V$, $i_1, ..., i_n$ and the vertices of $V'$, $j_1, ..., j_n$ such that for any $1 \leq k \leq |\mathcal{M}|$, $\{i_k,j_k\} \in \mathcal{M}$ and for any $1 \leq l < k \leq |\mathcal{M}|$, $\{i_k,j_l\} \notin E$.
\end{itemize}
\end{thm}

As a consequence of this theorem, we suppose w.l.o.g. that given any constrained matching $\mathcal{M}$ in a bipartite graph $(V,V',E)$, the vertices of $V$ and $V'$ are ordered so that for any $1 \leq k \leq |\mathcal{M}|$, $\{i_k, j_k\} \in \mathcal{M}$ and \begin{equation}\label{eq2} \text{ for any } 1 \leq l < k \leq |\mathcal{M}|, \ \ \{i_k, j_l\} \notin E \end{equation}

\medskip

For the proof of the following lemma, it may be useful to recall the definition of the bipartite graph $B_G = (V,V',E)$ associated with a loop directed graph $G$: the vertex sets $V = \{i_1, ..., i_n\}$ and $V' = \{j_1, ..., j_n\}$ are two copies of the vertex set of $G$. Moreover, $\{i_s,j_t\}$ is an edge of $B_G$ if and only if there is a directed edge from vertex $j_t$ to vertex $i_s$ in $G$.

\begin{lem}
\label{lem1}
Let $G$ be a loop directed graph with vertex set $V$,  $B_G = (V,V',E)$ its associated bipartite graph and $\mathcal{M} := \{ \{i_1,j_1\}, ..., \{i_t,j_t\}\}$ a constrained matching in $B_G$. Then, $V \backslash \{i_1, ..., i_t\}$ is a zero forcing set in $G$ with chronological list of forces $j_1 \rightarrow i_1, ..., j_t \rightarrow i_t$.
\end{lem}
\begin{proof}
Suppose that in the initial coloring of $G$, vertices $i_1, ..., i_t$ are the only white vertices.
From the definition of the bipartite graph $B_G$, we know that in $G$ vertex $i_1$ is a out-neighbour of vertex $j_1$, since $\{i_1,j_1\} \in E$.
Moreover, from property (\ref{eq2}), we know that in $G$ vertex $i_1$ is the only white out-neighbour of vertex $j_1$. Therefore, $j_1$ forces $i_1$.
By iterating this argument on all the edges $\{i_2,j_2\}, ..., \{i_t,j_t\}$ of $\mathcal{M}$, we prove that $V \backslash \{i_1, ..., i_t\}$ is a zero forcing set of $G$.
\end{proof}

\begin{thm}
\label{thm34}
Let $G$ be a loop directed graph with vertex set $V$ and $B_G = (V,V',E)$ the bipartite graph associated with $G$. Then,  $V \backslash \{i_1, ..., i_t\}$ is a zero forcing set of $G$ with a chronological list of forces $j_1 \rightarrow i_1, j_2 \rightarrow i_2, ..., j_t \rightarrow i_t$ if and only if $\mathcal{M} := \{ \{i_1,j_1\}, \{i_2,j_2\}, ..., \{i_t, j_t\}\}$ is a constrained matching in $B_G$.
\end{thm}
\begin{proof}
The sufficient condition has been proved in Lemma \ref{lem1}. Suppose that $V \backslash \{i_1, ..., i_t\}$ is a zero forcing set of $G$ with chronological list of forces $$j_1 \rightarrow i_1, ..., j_t \rightarrow i_t.$$ For any $1 \leq l \leq t$, since $j_l \rightarrow i_l$, vertex $i_l$ is a out-neighbour of vertex $j_l$. Therefore, $\{i_l, j_l\}$ is an edge of $B_G$. In addition, since any white vertex is forced only once and any vertex forces at most one vertex, $\mathcal{M} := \{ \{i_1,j_1\}, \{i_2,j_2\}, ..., \{i_t, j_t\}\}$ is a matching in $B_G$. Since $j_1$ forces $i_1$, vertex $i_1$ is the only white out-neighbour of vertex $j_1$. Hence, for any $1 < k \leq t$, $\{i_k,j_1\} \notin E$. By iterating this argument on all the forces $j_2 \rightarrow i_2, ..., j_t \rightarrow i_t$, we prove that $\mathcal{M}$ meets property $(\ref{eq2})$. Therefore, from Theorem \ref{thm2}, $\mathcal{M}$ is a constrained matching.
\end{proof}

Notice that in the previous theorem the set $V \backslash \{i_1, ..., i_t\}$ is the set of unmatched vertices in the vertex subset $V$ of $B_G$ resulting from the constrained matching $\mathcal{M}$.

\medskip

Thanks to the previous result, we can re-state Theorem \ref{thm42} in terms of zero forcing sets. 

\medskip

Remember that $V_{loop}$ is the set of vertices with a loop in the underlying loop directed graph $G$.

\begin{thm}[Restatement of Theorem \ref{thm42}]
\label{thm55}
Let $G$ be a loop directed graph on $n$ vertices with pattern $\textbf{A}$ and $S$ be an input set with cardinality $m \leq n$. System $(\textbf{A}, \textbf{B}(S))$ is strongly $S$-controllable if and only if 
\begin{enumerate}[-]
\item $S$ is a zero forcing set of $G$ and
\item $S$ is a zero forcing set of $G_\times$ for which there is a chronological list of forces that does not contain any force of the form $i \rightarrow i$ with $i \in V_{loop}$.
\end{enumerate}
\end{thm}
\begin{proof}
Theorem \ref{thm34} states that $\textbf{A}(S|.)$ has a constrained $(n-m)$-matching if and only if $S$ is a zero forcing set of $G$. Indeed, in the bipartite graph $(V,V',E)$ associated with $\textbf{A}$, the unmatched vertices of $V$ resulting from a constrained $(n-m)$-matching of $\textbf{A}(S|.)$ are the vertices in $S$.

From the same argument,  $\textbf{A}_\times(S|.)$ has a constrained $(n-m)$-matching if and only if $S$ is a zero forcing set of $G_\times$. Besides, there is a constrained $(n-m)$-matching which is $V_{loop}$-less in $\textbf{A}_\times(S|.)$ if and only if in $G_\times$ zero forcing set $S$ has a chronological list of forces with no force of the form $i \rightarrow i$ with $i \in V_{loop}$. Indeed, Theorem \ref{thm34} states that the edges in the constrained matching form a chronological list of forces for zero forcing set $S$.

Theorem \ref{thm42} concludes the proof.
\end{proof}

Notice that testing if a system is strongly $S$-controllable is equivalent to checking if $S$ is a zero forcing set in a loop directed graph, which can be done in time $\mathcal{O}(n^2)$. As an example, consider the system whose underlying loop directed graph $G$ is in Figure 1 and check that this system is strongly $S$-controllable for $S = \{1\}$. We immediately check that $S$ is a zero forcing set of $G$ and that $S$ is a zero forcing set of $G_\times$ with chronological list of forces: $1 \rightarrow 2, 2 \rightarrow 3$. Since $V_{loop} = \{1\}$ and in this list vertex 1 does not force itself, the system is then strongly $S$-controllable.

\begin{figure}
\centering
\includegraphics[width=0.4\textwidth]{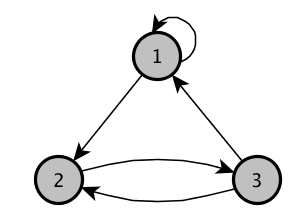}
\caption{A system with this loop directed graph as underlying structure is strongly $S$-controllable with $S = \{1\}$.}
\label{fig1}
\end{figure}

\medskip

The following lemma is a first step to a statement of Theorem \ref{thm43} in terms of zero forcing sets in a \textit{simple} directed graph.

\begin{lem}
\label{lem56}
Consider a loop directed graph $G$ on $n$ vertices underlying a self-damped system. A vertex subset $S \subseteq V$ is a input set for strong controllability of the system if and only if $S$ is a zero forcing set in $G_\times$ for which there is a chronological list of forces with no force of the form $i \rightarrow i$.
\end{lem}
\begin{proof}
Since the system is self-damped, $G_\times = G$ and $V_{loop} = V$. Consequently, this result follows from Theorem \ref{thm55}.
\end{proof}

From this lemma, we deduce a statement of Theorem \ref{thm43} in terms of zero forcing sets in a \textit{simple} directed graph. The simple directed graph $G_s$\textcolor{red}{\footnote{Whenever $s$ is used as a subscript, it always refers to a simple directed graph. The vertex subset $S$ is never used as a subscript.}} is obtained from the loop directed graph $G$ by removing the loops on its vertices.

Recall that in a simple directed graph, a vertex must be black to be able to force one of its out-neighbors.

\begin{thm}[Restatement of Theorem \ref{thm43}]
\label{thm57}
Consider a loop directed graph $G$ on $n$ vertices underlying a self-damped system. A vertex subset $S \subseteq V$ is a (minimum-size) input set for strong controllability of the system if and only if $S$ is a (minimum) zero forcing set in the simple directed graph $G_s$.
\end{thm}
\begin{proof}
We show that $S$ is a zero forcing set in $G_s$ if and only if $S$ is a zero forcing set in $G_\times$ for which there is a chronological list of forces with no force of the form $i \rightarrow i$. Indeed, in $G_\times$ any vertex has a loop. Moreover, if there is no force of the form $i \rightarrow i$, it means that any vertex must be black to be able to force one of its out-neighbours. Therefore, $S$ is a zero forcing set in $G_s$ if and only if $S$ is a zero forcing set in $G_\times$ for which there is a chronological list of forces with no force of the form $i \rightarrow i$.

\medskip

This observation together with Lemma \ref{lem56} proves the theorem.
\end{proof}

This shows that the minimum zero forcing sets in the simple directed graph $G_s$ provide minimum-size input sets for strong controllability of the underlying self-damped system. However, we deduce from the NP-hardness result of Aazami \cite{aazami} about the zero forcing number of a simple undirected graph that computing a minimum zero forcing set in a simple directed graph is an NP-hard problem. Nevertheless, there are some efficient algorithms computing minimum zero forcing sets in a simple tree, see for example Algorithm 2.3 in \cite{fallat}. Such algorithms together with the previous theorem show that one can compute efficiently a minimum-size input set for strong controllability of a self-damped system with a tree structure.

\begin{thm}
One can compute in polynomial time a minimum-size input set for strong controllability of a self-damped system with a tree-structure.
\end{thm}

\section{Related results}

Recently, a result related to Theorems \ref{thm55} and \ref{thm57} has been found by Monshizadeh \textit{et al.} \cite{monsh}, but for systems underlying a \textit{simple} directed graph. In this section, we explain the main result of \cite{monsh} and compare it with our results from the previous section.

\medskip

A simple directed graph $G_s = (V,E)$ on $n$ vertices defines the matrix family $$\mathcal{Q}_{sd}(G_s) = \{ A \in \mathbb{R}^{n \times n} : \text{ for } i \neq j, a_{ij} \neq 0 \Leftrightarrow (j,i) \in E\}.$$ The pattern $\textbf{A}_s$ associated with $G_s$ is an $n \times n$ matrix,  where any off-diagonal entry $a_{ij}$ is a star $\star$ if and only if $(j,i)$ is an edge of $G_s$, every diagonal entry is a question mark and the other entries are zero. As an example, the pattern $\textbf{A}_s$ of the simple directed graph $G_s$ in Figure 2b) is $$\textbf{A}_s = \left(\begin{array}{ccc} ? & \star & 0 \\ \star & ? & 0 \\ \star & \star & ?\end{array}\right).$$ A star $\star$ denotes a nonzero entry, whereas a question mark can be a zero or nonzero entry. The following matrices $$A_1 = \left(\begin{array}{ccc} -3 & 1 & 0 \\ 9 & 0 & 0 \\ -5 & -4 & 0 \end{array}\right) \ \ \ A_2 = \left(\begin{array}{ccc} 0 & 1 & 0 \\ 2 & -3 & 0 \\ 1 & -4 & 8 \end{array}\right)$$ are both in $\mathcal{Q}_{sd}(G_s)$ and are two realizations of pattern $\textbf{A}_s$. We write $A_1 \in \textbf{A}_s$ and $A_2 \in \textbf{A}_s$.

\medskip

Unlike the matrices of $\mathcal{Q}_{ld}(G)$ defined by a loop directed graph $G$, the matrices in $\mathcal{Q}_{sd}(G_s)$ have a free diagonal. For an example, have a look at the loop directed graph $G$ in Figure 2a) whose associated simple directed graph is in Figure 2b). The pattern $\textbf{A}$ associated with $G$ is $$\textbf{A} = \left(\begin{array}{ccc} \star & \star & 0 \\ \star & 0 & 0 \\ \star & \star & 0  \end{array}\right).$$ Therefore, while $A_1$ and $A_2$ are both a realization of $\textbf{A}_s$, $A_1$ is a realization of $\textbf{A}$ whereas $A_2$ is not.

\medskip

Consequently, given a loop directed graph $G$ and its associated simple directed graph $G_s$, $$\mathcal{Q}_{ld}(G) \subseteq \mathcal{Q}_{sd}(G_s).$$

\begin{figure}
\centering
\includegraphics[width=0.8\textwidth]{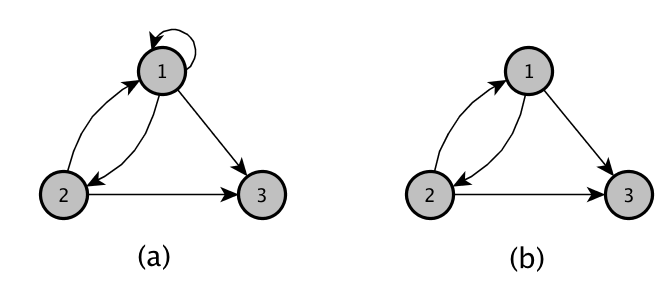}
\caption{A loop directed graph $G$ in (a) and its associated simple directed graph $G_s$ in (b). Given the input set $S = \{1\}$, system $(\textbf{A}, \textbf{B}(S))$ underlying $G$ is strongly $S$-controllable, whereas system $(\textbf{A}_s, \textbf{B}(S))$ underlying $G_s$ is not.}
\label{fig1}
\end{figure}

A networked system whose underlying graph is a \textit{simple} directed graph $G_s = (V,E)$ on $n$ vertices is a linear system of the form: $$\dot{x}(t) = Ax(t) + Bu(t),$$ where $x(t) \in \mathbb{R}^n$ is the state vector, $u(t) \in \mathbb{R}^m$ is the outside controller, $A \in \mathcal{Q}_{sd}(G_s)$ and $B \in \textbf{B}(S)$, for some vertex set $S \subseteq V$ of size $m$. Such a system is referred to as system $(A,B)$.

\medskip

Given an input set $S$, any system $(A,B)$ underlying a simple directed graph $G_s$ with pattern $\textbf{A}_s$ is a realization of the pair $(\textbf{A}_s, \textbf{B}(S))$, meaning that $A \in \textbf{A}_s$ and $B \in \textbf{B}(S)$. We write $(A,B) \in (\textbf{A}_s, \textbf{B}(S))$. The set of systems underlying $G_s$ is referred to as system $(\textbf{A}_s, \textbf{B}(S))$.

\medskip

Similarly to the systems whose underlying graph is a loop directed graph, we say that a system $(\textbf{A}_s, \textbf{B}(S))$ is strongly $S$-controllable if all the systems $(A,B) \in (\textbf{A}_s, \textbf{B}(S))$ are controllable.

\medskip

Here is the main result of \cite{monsh} about strong controllability of system $(\textbf{A}_s, \textbf{B}(S))$.

\begin{thm}\cite{monsh}
\label{monshi}
Given a simple directed graph $G_s = (V,E)$ with pattern $\textbf{A}_s$ and an input set $S \subseteq V$, system $(\textbf{A}_s, \textbf{B}(S))$ is strongly $S$-controllable if and only if $S$ is a zero forcing set of $G_s$. 
\end{thm}

For a subset $\mathcal{Q}'(G_s) \subseteq \mathcal{Q}_{sd}(G_s)$, a system subset of $(\textbf{A}_s, \textbf{B}(S))$ is the set of systems $(A,B)$ where $A \in \mathcal{Q}'(G_s)$ and $B \in \textbf{B}(S)$.

Such a system subset is strongly $S$-controllable if all the systems $(A,B)$ with $A \in \mathcal{Q}'(G_s)$ and $B \in \textbf{B}(S)$ are controllable.

In particular, if $G_s$ is the simple directed graph with pattern $\textbf{A}_s$ associated with a loop directed graph $G$ with pattern $\textbf{A}$, the system set $(\textbf{A}, \textbf{B}(S))$ is a subset of the systems in $(\textbf{A}_s, \textbf{B}(S))$ since $\mathcal{Q}_{ld}(G) \subseteq \mathcal{Q}_{sd}(G_s)$.

\medskip

By definition, if system $(\textbf{A}_s, \textbf{B}(S))$ is strongly $S$-controllable, then any system subset of $(\textbf{A}_s, \textbf{B}(S))$ is strongly $S$-controllable. However, the converse is false in general.

As an example, take input set $S = \{1\}$ and consider system $(\textbf{A}, \textbf{B}(S))$ whose underlying graph is the loop directed graph $G$ in Figure 2a) and system $(\textbf{A}_s, \textbf{B}(S))$ whose underlying graph is the simple directed graph $G_s$ in Figure 2b) associated with $G$.
From Theorem \ref{thm55}, we deduce that $(\textbf{A}, \textbf{B}(S))$ is strongly $S$-controllable since $S$ is a zero forcing set of $G$ and a zero forcing set of $G_\times$ where vertex 1 does not need to force itself. Instead, since $S$ is not a zero forcing set in the simple directed graph $G_s$, Theorem \ref{monshi} claims that $(\textbf{A}_s, \textbf{B}(S))$ is not strongly $S$-controllable.

\medskip

Consequently, given a loop directed graph $G$ with pattern $\textbf{A}$ and its associated simple directed graph $G_s$ with pattern $\textbf{A}_s$, Theorem \ref{thm55} analyzes the strong controllability of system $(\textbf{A}, \textbf{B}(S))$, for some input set $S$ whereas Theorem \ref{monshi} is about the strong controllability of a bigger system set $(\textbf{A}_s, \textbf{B}(S))$ underlying $G_s$.

In addition, from Theorems \ref{thm57} and \ref{monshi}, we deduce the following result. 

\begin{thm}
Let $G$ be a loop directed graph with pattern $\textbf{A}$ underlying a \textbf{self-damped} system and $G_s$ be its associated simple directed graph with pattern $\textbf{A}_s$. System $(\textbf{A}, \textbf{B}(S))$ is strongly $S$-controllable if and only if system $(\textbf{A}_s, \textbf{B}(S))$ is strongly $S$-controllable.
\end{thm}

\medskip

Other system subsets of $(\textbf{A}_s, \textbf{B}(S))$ were studied in \cite{monsh}, notably when the simple directed graph is symmetric. We refer the reader to Section IV.A of \cite{monsh} for more details.

\section{Conclusion}

This paper links the notions of zero forcing, constrained matching and strong controllability.

\medskip

As a first result, we have shown that computing the zero forcing number of any \textit{loop directed} graph is NP-hard. This completes the NP-hardness result of \cite{aazami} about the zero forcing number of a \textit{simple undirected} graph.

\medskip

The rest of the paper sheds a new light on the strong controllability of a networked system, through the zero forcing sets. Our results are based on a one-to-one correspondance between the zero forcing sets in a loop directed graph $G$ and the constrained matchings in the bipartite graph associated with $G$. We have re-stated some results of \cite{chapman} about strong controllability in terms of zero forcing sets. On the one hand, these new statements show that testing whether or not a system is strongly $S$-controllable from an input set $S$ is equivalent to checking if $S$ is a zero forcing set in a loop directed graph. On the other hand, they show that the (minimum) zero forcing sets in the \textit{simple} interconnection graph provide (minimum-size) input sets for the strong controllability of a self-damped system. In particular, we deduce that one can find in polynomial time a minimum-size input set for the strong controllability of a self-damped system with a tree structure, using existing algorithms on zero forcing. 
\medskip

A similar work was done in \cite{monsh} when the underlying graph is a \textit{simple} directed graph and in\cite{burgarth3, burgarth, burgarth2, burgarth4} where the link between the zero forcing sets of a simple undirected graph and the controllability of a quantum system has been shown.

\medskip

All these results show the role of the zero forcing sets in the study of the dynamics of networked systems and should motivate additional research on zero forcing. Here are some open problems:

\begin{enumerate}[-]
\item The NP-hardness of the computation of the zero forcing number of any simple undirected graph implies the NP-hardness for the zero forcing number of any simple directed graph. In this paper, we have proved that computing the zero forcing number of any loop directed graph is also NP-hard. We also expect NP-hardness for the zero forcing number of any loop undirected graph. However, a thorough argument is still needed.

\item Theorem \ref{thm57} shows that selecting a minimum-size input set for strong controllability in a self-damped system is equivalent to finding a minimum zero forcing set in a simple directed graph. However, finding a minimum zero forcing set in a simple directed graph is known to be NP-hard. Nevertheless, for particular graphs the problem has turned out to be easy, for example if the graph is a symmetric path, or the complete graph, or a tree. Can we identify more general simple directed graphs for which a minimum zero forcing set can be selected in polynomial time ? Is there a polynomial time algorithm if the graph is a simple \textit{directed tree} ?
\end{enumerate}

\medskip
\noindent
\small{\textbf{ACKNOWLEDGMENTS.} The authors acknowledge support from the Belgian Programme of Interuniversity Attraction Poles and an Action de Recherche Concert\'ee (ARC) of the Wallonia-Brussels Federation.}

\newpage

%
%
%
%
%
%
%
%
%
%
%
%
%

\newpage

\noindent
\textbf{Appendix A}

\medskip
\noindent
It turns out that the following result from [9, Corollary 11] is incorrect.

\begin{cor}\cite{chapman}
\label{cor81}
Consider an undamped system whose underlying graph has pattern $\textbf{A}$ (undamped = diagonal of $\textbf{A}$ is zero). In the bipartite graph $(V,V',E)$ associated with $\textbf{A}_\times$, denote by $S \subseteq V$ and $S' \subseteq V'$ the unmatched vertices resulting from a maximum constrained self-less matching of $\textbf{A}_\times$. Then, $S$ is a minimum-size input set for the strong controllability of the system.
\end{cor}

\noindent
Here is a counter-example to this result: the loop directed graph underlying the system is in Figure 3a). Patterns $\textbf{A}$ and $\textbf{A}_\times$ are: $$\textbf{A} = \left(\begin{array}{ccc} 0 & \star & 0 \\ \star & 0 & 0 \\ \star & \star & 0 \end{array}\right) \ \ \ \ \textbf{A}_\times = \left(\begin{array}{ccc} \star & \star & 0 \\ \star & \star & 0 \\ \star & \star & \star \end{array}\right)$$  

\begin{figure}
\centering
\includegraphics[width=0.99\textwidth]{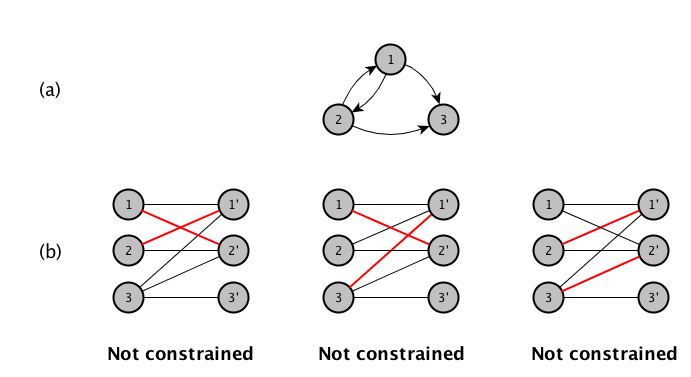}
\caption{(a) Loop directed graph with pattern $\textbf{A}$ underlying an undamped system - (b) The bipartite graph associated with pattern $\textbf{A}_\times$ and its three self-less 2-matchings. None of them is constrained.}
\label{fig1}
\end{figure}

\noindent
Figure 3b) shows that none of the three self-less 2-matchings in $\textbf{A}_\times$ is constrained. Therefore, a maximum constrained self-less matching in $\textbf{A}_\times$ has size 1. According to the previous corollary, the minimum size of an input set for strong controllability is 2.

\medskip
\noindent
However, if we apply Theorem \ref{thm42} to our example, we find that the system is strongly $S$-controllable from $S = \{1\}$. Indeed, since the system is undamped, $V_{loop} = \emptyset$ and we easily check that $\textbf{A}(S|.)$ and $\textbf{A}_\times(S|.)$ have both a constrained 2-matching.

\medskip
\noindent
Consequently, Theorem \ref{thm42} and Corollary \ref{cor81} are in contradiction. We have revised the proofs of the two results. Theorem \ref{thm42} has turned out to be correct whereas an error has been depicted in the proof of Corollary \ref{cor81}. Indeed, Corollary \ref{cor81} considers a maximum constrained SELF-LESS matching in $\textbf{A}_\times$. The self-less condition is essential to deduce that $\textbf{A}(S|.)$ and $\textbf{A}_\times(S|.)$ have both a constrained matching of size $n-|S|$. However, Theorem \ref{thm42} applied to an undamped system claims that the system is strongly $S$-controllable if and only if $\textbf{A}(S|.)$ and $\textbf{A}_\times(S|.)$ have both a $(n-|S|)$-constrained matching. Nevertheless, it may occur that these two matchings are different and that the one of $\textbf{A}_\times(S|.)$ is not self-less. That is what happens in our counter-example. Consequently, the maximum constrained self-less matching in Corollary \ref{cor81} may provide an input set for strong controllability whose size is not minimum. A correct statement of Corollary \ref{cor81} should be: 

\begin{cor}[Correct statement of Corollary \ref{cor81}] Consider an undamped system whose underlying graph has pattern $\textbf{A}$. In the bipartite graph $(V,V',E)$ associated with $\textbf{A}_\times$, denote by $S \subseteq V$ and $S' \subseteq V'$ the unmatched vertices resulting from a maximum constrained self-less matching of $\textbf{A}_\times$. Then, $S$ is an input set for the strong controllability of the system.
\end{cor}

\end{document}